\documentclass[11pt]{article}
\usepackage{snapshot}
\usepackage{fullpage}
\usepackage{color}
\usepackage{tikz}
\usepackage{subcaption}
\usepackage{booktabs}
\usepackage[linesnumbered,ruled,vlined]{algorithm2e}
\usepackage{algpseudocode}
\SetKwInOut{Input}{Input}
\SetKwInOut{Output}{Output}
\usepackage{amsfonts}
\usepackage{amssymb}
\usepackage{amsthm}
\usepackage{amsmath}

\newcommand{\True}{\mathrm{True}}

\newcommand{\CA}{\ensuremath{\mathcal{C}_A}}
\newcommand{\CB}{\ensuremath{\mathcal{C}_B}}
\newcommand{\comp}{\mathrm{comp}}
\newcommand{\ctr}{\mathrm{center}}

\usepackage{ifpdf}
\newcommand{\mydriver}{hypertex}
\ifpdf
\renewcommand{\mydriver}{pdftex}
\fi
\usepackage[breaklinks,\mydriver]{hyperref}

\usepackage[capitalise]{cleveref}
\usepackage{footnote}

\makesavenoteenv{tabular}
\makesavenoteenv{algorithm}

\newcommand{\E}{\mathop{\mathbf{E}}}

\newcommand{\p}{\mathbf{p}}

\newcommand{\q}{\mathbf{q}}

\newcommand{\x}{\mathbf{x}}

\newtheorem{theorem}{Theorem}[section]

\newtheorem{lemma}[theorem]{Lemma}

\newtheorem{claim}[theorem]{Claim}

\theoremstyle{definition}
\newtheorem{definition}[theorem]{Definition}

\title{An Optimal Separation between Two Property Testing Models for Bounded Degree Directed Graphs\footnote{Supported in part by NSFC grant 62272431 and ``the Fundamental Research Funds for the Central Universities''.}}
\author{
	Pan Peng\footnote{
	School of Computer Science and Technology, University of Science and Technology of China.  Email: \href{ppeng@ustc.edu.cn}{ppeng@ustc.edu.cn}}
	\and
	Yuyang Wang\footnote{School of Computer Science and Technology, University of Science and Technology of China.  Email: \href{wangyvyang@mail.ustc.edu.cn}{wangyvyang@mail.ustc.edu.cn}}
}

\begin{document}
	\date{}
	\maketitle
\begin{abstract}
	We revisit the relation between two fundamental property testing models for bounded-degree \emph{directed} graphs: the \emph{bidirectional} model in which the algorithms are allowed to query both the outgoing edges and incoming edges of a vertex, and the \emph{unidirectional} model in which only queries to the outgoing edges are allowed. Czumaj, Peng and Sohler [STOC 2016] showed that for directed graphs with both maximum indegree and maximum outdegree upper bounded by $d$, any property that can be tested with query complexity $O_{\varepsilon,d}(1)$ in the bidirectional model can be tested with $n^{1-\Omega_{\varepsilon,d}(1)}$ queries in the unidirectional model. In particular, {if the proximity parameter $\varepsilon$ approaches $0$, then the query complexity of the transformed tester in the unidirectional model approaches $n$}. It was left open if this transformation can be further improved or there exists any property that exhibits such an extreme separation.
	
	We prove that testing \emph{subgraph-freeness} in which the subgraph contains $k$ source components,  requires $\Omega(n^{1-\frac{1}{k}})$ queries in the unidirectional model. This directly gives the first explicit properties that exhibit an $O_{\varepsilon,d}(1)$ vs $\Omega(n^{1-f(\varepsilon,d)})$ separation of the query complexities between the bidirectional model and unidirectional model, where $f(\varepsilon,d)$ is a function that approaches $0$ as $\varepsilon$ approaches $0$. Furthermore, our lower bound also resolves a conjecture by Hellweg and Sohler [ESA 2012] on the query complexity of testing $k$-star-freeness. 
\end{abstract}

\section{Introduction}
Graph property testing is a framework for studying extremely fast (randomized)  algorithms for solving a relaxation of classical decision problems on graphs. Given a graph property $P$, we are interested in designing an algorithm, called a \emph{property tester}, that with high constant probability, accepts any graph $G$ that satisfies $P$,  and rejects any graph that is ``far'' from satisfying $P$, i.e., one needs to modify a significant fraction of the representation (e.g., adjacency matrix or adjacency list) of the graph to make it satisfy $P$. It is assumed that the algorithm is given oracle access to the representation of the graph and the goal of a property tester is to solve the above problem by making as few queries to the oracle as possible. Since the seminal works by Rubinfeld and Sudan \cite{rubinfeld1996robust} (on algebraic property testing) and Goldreich, Goldwasser and Ron \cite{goldreich1998property} (on combinatorial and graph property testing), a lot of efforts have been made on studying which properties can be tested within a sublinear (e.g., constant) number of queries in several classical models, e.g., the dense graph model \cite{goldreich1998property,alon2009combinatorial} and bounded-degree graph model \cite{GoldreichRon2002}. In particular, we have see a rapid development of property testing on \emph{undirected} graphs in the past two decades. We refer to the recent book \cite{goldreich2017introduction} for a survey. 

On the other hand, we still do not know much about property testing in \emph{directed graphs (digraphs)} so far. Bender and Ron \cite{bender2002testing} introduced two fundamental models for studying %
directed graph property testing. The first is called \emph{bidirectional} model, where the algorithm is allowed to query both outgoing and incoming edges of a vertex; the second is called \emph{unidirectional} model, where the algorithm is only allowed to query the outgoing edges, while not incoming edges. The latter model seems more realistic for some applications. For example, consider the webgraphs. It is much easier to query the outgoing edges $(u,v)$ (which corresponds to a hyperlink from webpage $u$ to webpage $v$) than querying the incoming edges. In this paper, we focus on bounded-degree directed graphs. A digraph $G$ is said to be \emph{$d$-bounded}, if both the maximum outdegree and maximum indegree of $G$ are upper bounded by $d$, which is assumed to be a constant. 

Bender and Ron gave an algorithm for testing strong connectivity with $\tilde{O}(1/\varepsilon)$ queries in the bidirectional model, and showed that there is a lower bound of $\Omega(\sqrt{n})$ queries for any algorithm with two-sided error\footnote{A tester for a property $P$ is said to have \emph{one-sided error} if it accepts every (di)graph satisfying $P$, and it errs if the graph is far from having $P$. It is said to have \emph{two-sided error} if it errs in both cases.} in the unidirectional model. Goldreich \cite{goldreich2010introduction}, and Hellweg and Sohler \cite{hellweg2012property} gave a lower bound of $\Omega(n)$ queries for testing strong connectivity with one-sided error in the unidirectional model. The works \cite{goldreich2010introduction,hellweg2012property} also gave testers for strong connectivity with $n^{1-1/(d+1/\varepsilon)}$ queries with two-sided error in the unidirectional model. In \cite{hellweg2012property}, the authors gave testers for subgraph-freeness with $O(n^{1-\frac{1}{k}})$ queries in the unidirectional model, where $k$ is the number of connected components in the subgraph that have no incoming edges. It is known that a few properties can be tested with a constant number of queries in the bidirectional model, including Eulerianity \cite{orenstein2011testing}, $k$-edge connectivity \cite{orenstein2011testing,yoshida2010testing,forster2020computing}, $k$-vertex connectivity \cite{orenstein2011testing,forster2020computing}.

Towards a deeper understanding of testing properties of bounded degree
directed graphs (digraphs), Czumaj, Peng and Sohler \cite{czumaj2016relating} studied the relation between these two models and provided a generic transformation that
converts testers with constant query complexity in the bidirectional model, to testers with
sublinear query complexity in the unidirectional model. Specifically, in \cite{czumaj2016relating}, it was shown that any property $P$ that can be tested with\footnote{Throughout the paper, we use the notation $O_{\varepsilon,d}()$ (resp. $\Omega_{\varepsilon,d}()$) to describe a function in the Big-O (resp. Big-Omega) notation assuming that $\varepsilon$ and $d$ are constant.} $q=O_{\varepsilon,d}(1)$ queries in the bidirectional model can be tested with $n^{1-d^{-\Theta(q)}}=n^{1-\Omega_{\varepsilon,d}(1)}$ queries in the unidirectional model (with two-sided error). In particular, {if the proximity parameter $\varepsilon$ approaches $0$, then the query complexity of the transformed tester in the unidirectional model approaches $n$ (as the term $d^{-\Theta(q)}$ approaches $0$)}.

One natural question that is left open is that \emph{is the above the transformation tight}? Or equivalently, \emph{can we achieve a much better query complexity, say $n^{c-\Omega_{\varepsilon,d}(1)}$, in the latter model, for some universal constant $c<1$}? Indeed, currently, the best known lower bound for this transformation is for testing \emph{$3$-star-freeness}, where a $3$-star is a $4$-vertex directed graph  such that there exists one center vertex $v$, and for any other three vertices $u$, there is an edge from $u$ to $v$, and no other edges exist.  Hellweg and Sohler \cite{hellweg2012property} have shown that $3$-star-freeness can be tested with a constant number of queries in bidirectional model, while the query complexity of a tester for this property in the unidirectional model is $\Theta(n^{2/3})$ for any constant $\varepsilon>0$. Therefore, there is still a significant gap between the upper bound (i.e., $n^{1-\Omega_{\varepsilon,d}(1)}$) in the bidirectional model in the transformation and the current best-known lower bound (i.e., $\Omega(n^{2/3})$). 

Before we state our result, we formally introduce the definition of property testing in both directional and unidirectional models.  
Let $P=(P_n)_{n\in \mathbb{N}}$ be a $d$-bounded digraph property, where $P_n$ is a property of $d$-bounded digraphs with $n$ vertices. An $n$-vertex graph $G$ is said to be $\varepsilon$-far from satisfying $P_n$ if one needs to modify more than $\varepsilon d n$ edges to make it a $d$-bounded digraph with property $P_n$, where $\varepsilon>0$ is called the proximity parameter.  We say that $P$ is \emph{$q$-query testable} (or that $P$ \emph{can be tested with query complexity $q$})  %
if for every $n$, $\varepsilon$ and $d$, there exists a tester that makes $q = q(n, \varepsilon, d)$ queries and with probability at least $\frac23$, accepts any $n$-vertex $d$-bounded digraph $G$ satisfying $P$,  and rejects any $n$-vertex $d$-bounded digraph $G$ that is $\varepsilon$-far from satisfying $P$. We call such a tester an \emph{$\varepsilon$-tester} for $P$. %

We show that there exists a property that exhibits an $O_{\varepsilon,d}(1)$ vs $\Omega(n^{1-\Theta_{\varepsilon,d}(1)})$ separation of the query complexities between the bidirectional model and unidirectional model,  which implies that the transformation of \cite{czumaj2016relating} is essentially tight.  
\begin{theorem}\label{thm:main}
	For any sufficiently small constant $\varepsilon>0$, there exists a digraph property $P = P_{\varepsilon,d}$ such that $P$ can be tested with $O_{\varepsilon,d}(1)$ queries in the bidirectional model, while any $\varepsilon$-tester for $P$ in the unidirectional model requires $n^{1-f(\varepsilon,d)}$ queries, where $f(\varepsilon,d)$ is a function that approaches $0$ when $\varepsilon$ approaches $0$. 
\end{theorem}

The above theorem is a direct corollary from the following result regarding testing subgraph-freeness. %
Let $H$ be a directed graph. A strongly connected component\footnote{We call $W\subseteq V(H)$ a \emph{strongly connected component} of $H$ if the subgraph $H[W]$ of $H$ induced by $W$ is strongly connected, and there does not exist any set of vertices $X \subseteq V(H)\setminus W$ such that the subgraph of $H$ induced by $W \cup X$ is strongly connected. That is, the subgraph of $H[W]$ is a strongly connected and maximal. } $W$ is called a \emph{source component} of $H$, if there is no edge from $V(H) \setminus W$ to $W$. %
A directed graph $H$ is said to be \emph{weakly connected} if its underlying undirected graph (i.e., the graph that is obtained by ignoring all the directions of the edges) is connected. 
For example, we note that $k$-star is just a weakly connected digraph with $k$ source components, where a directed graph $H$ with $k+1$ vertices is called a \emph{$k$-star} if there is a vertex $v$ such that each of the other $k$ vertices has exactly one edge pointing to $v$, and $H$ does not contain any other edges. Let $G$ and $H$ be two directed graphs.  The graph $G$ is said to be \emph{$H$-free} if $H$ does not appear as a subgraph in $G$.  %
We have the following theorem on testing $H$-freeness for any (constant-size) $H$ with $k$ source components. %
\begin{theorem}\label{thm:subgraph-lowerbound}
	Let $k$ be any integer such that $k\geq 2$. Let $d$ be some constant. Let $H$ be a weakly connected\footnote{For graphs $H$ that is not weakly connected, we can handle each of its weakly connected components separately.} directed graph with $k$ source components. %
	There exists an $\varepsilon_0=\Theta_{d,k}(1)$ such that any $\varepsilon_0$-tester for testing $H$-freeness %
	of an $n$-vertex $d$-bounded graph requires at least $\Omega(n^{1-\frac{1}{k}})$ queries in the  unidirectional model.
\end{theorem}

We remark that it has been shown by Hellweg and Sohler \cite{hellweg2012property} that for any $H$ with $k$ source components, $H$-freeness can be tested with query complexity $O_{\varepsilon,d,k}(1)$ in the bidirectional model, and also can be tested with query complexity $O_{\varepsilon,d,k}(n^{1-\frac{1}{k}})$ in the unidirectional model\footnote{{On the high level, their algorithms use the following observation: if a bounded-degree directed graph $G$ is $\varepsilon$-far from $H$-freeness, then $G$ contains $\Omega(\varepsilon n)$ vertex-disjoint copies of $H$. Then in the bidirectional model, one can sample a constant number of vertices and perform BFS from each sampled vertex to find a copy of $H$; in the unidirectional model, one can sample many edges to see if some copy of $H$ is formed.}}. %
Given the above result, we can easily prove Theorem \ref{thm:main}.
\begin{proof}[Proof of Theorem \ref{thm:main}]
	Let $\eta>0$ and define property $P_\eta$ to be the property of being $H$-free, for any $H$ that is weakly connected digraph with $k=\lceil{1/\eta}\rceil$ source components. According to Theorem \ref{thm:subgraph-lowerbound}, any $\varepsilon_0$-tester for $P_\eta$ requires at least $\Omega(n^{1-\frac{1}{k}})$ $>$ $\Omega(n^{1-\eta})$ queries in the unidirectional model, where $\varepsilon_0=\varepsilon_0(d,\eta)$ is a function of $d,\eta$. Now given any sufficiently small constant $\varepsilon>0$, let $\eta'$ be a number satisfying that $\varepsilon=\varepsilon_0(d,\eta')$. Then Theorem \ref{thm:main} follows by taking $P=P_{\eta'}$ and $f(\varepsilon,d)=\eta'$. 
\end{proof}

Furthermore, it was conjectured in \cite{hellweg2012property} that testing $k$-star-freeness requires $\Omega(n^{1-\frac{1}{k}})$ queries in the unidirectional model. Since $k$-star is a directed subgraph with $k$-source components, our Theorem \ref{thm:subgraph-lowerbound} resolves this conjecture.

\subsection{Discussions of previous ideas and our techniques} We first sketch the main ideas of the lower bound for testing $3$-star-freeness given by Hellweg and Sohler \cite{hellweg2012property}. Their proof makes use of a problem called \emph{testing $3$-occurrence-freeness}\footnote{In \cite{hellweg2012property}, the same problem was called {$3$-value freeness}.} of a sequence\footnote{{We use ``sequence'' rather than ``multiset'' as the position of each element affects our construction.}}. Let $A$ be a length-$n$ sequence of integers such that each element in $A$ is from $[\ell]:=\{1,\cdots,\ell\}$ and occurs at most $3$ times. We say $A$ is \emph{$3$-occurrence-free} if no integer in $A$ occurs exactly $3$ times in $A$. We say $A$ is \emph{$\varepsilon$-far from being $3$-occurrence-free} if one needs to change\footnote{It is allowed to use integers that are larger than $\ell$ to change the elements of $A$.} more than $\varepsilon n$ elements of $A$ to obtain a $3$-occurrence-free sequence. \cite{hellweg2012property} gave a local reduction from the problem of testing $3$-occurrence-freeness of a sequence to the problem of testing $3$-star-freeness. That is, given an instance $A$ with $m$ elements of $3$-occurrence-freeness, they constructed a graph $G$ with $\Theta(m)$ vertices, such that
\begin{itemize}
	\item[1)] if $A$ is $3$-occurrence-free, then $G$ is $3$-star-free; if $A$ is $\varepsilon$-far from being $3$-occurrence-free then $G$ is $\Theta(\varepsilon)$-far from being $3$-star-free; 
	\item[2)]  every query to $G$ can be answered by performing $O(1)$ queries to $A$.
\end{itemize} %

To obtain a lower bound for testing $3$-occurrence-freeness,  \cite{hellweg2012property} constructed two classes $\CA$, $\CB$ of length-$n$ sequences such that $\CA$ is a class of $3$-occurrence-free sequences and $\CB$ is a class of sequences that are $\Omega(1)$-far from being $3$-occurrence-free, and the \emph{frequency variables}, denoted by $X_A$ and $X_B$, of the sequences from these two different classes have $2$ proportional moments, i.e.,  
$$\frac{\E[X_B]}{\E[X_A]} = \frac{\E[X_B^2]}{\E[X_A^2]}. %
$$
Then the lower bound $\Omega(n^{2/3})$ for testing $3$-occurrence-freeness follows from a lower bound for distinguishing random variables with $2$-proportional moments given in \cite{raskhodnikova2009strong}. 

Now we note that to obtain a lower bound for testing $H$-freeness for any $H$ with $k$ source components, it suffices to give a lower bound for testing $k$-occurrence-freeness for general $k$ in the way similar as above. That is, we construct two classes $\CA$, $\CB$ of length-$n$ sequences such that $\CA$ is a class of $k$-occurrence-free sequences and $\CB$ is a class of sequences that are $\Omega_k(1)$-far from being $k$-occurrence-free, and the \emph{frequency variables}, denoted by $X_A$ and $X_B$, of the sequences from these two different classes have $k-1$ proportional moments, i.e., 
$$\frac{\E[X_B]}{\E[X_A]} = \frac{\E[X_B^2]}{\E[X_A^2]} = \cdots = \frac{\E[X_B^{k-1}]}{\E[X_A^{k-1}]}. 
$$

However, the main difficulty is to construct two classes of sequences satisfying the above equations for general $k\geq 3$, which was also pointed out in \cite{hellweg2012property}.  Besides the aforementioned construction in \cite{hellweg2012property} which only works for $k=3$, we also note that in \cite{raskhodnikova2009strong}, a special pair of random variables with $k-1$ proportional moments is also constructed (for establishing their lower bound for \textsc{Distinct-Elements}). That is, their random variables take values of the form $(B+3)^i$, for any integers $B>1, k>1$ and $i=0,\dots, k-1$. This leads to a large gap between the expectations of the corresponding variables. To show a lower bound for testing $k$-occurrence-freeness, we need to construct random variables taking values $1,2,\dots,k$, for any integer $k>1$. This is more challenging as it corresponds to a much smaller gap (which is arbitrarily close to $1$) between the expectations of the corresponding variables (see Lemma \ref{lemma:definitionsofpq}). 
To construct such two random variables, we establish some identities related to binomial coefficients, and use them to define two distributions satisfying a number of linear equations which in turn are necessary conditions for two variables having proportional moments.

We then give a local reduction from testing $k$-occurrence-freeness to testing $H$-freeness for $H$ with $k$-source components. The reduction also non-trivially generalizes the one for $3$-star-free in  \cite{hellweg2012property}, as $3$-star is a special subgraph with a nice symmetric property, while an arbitrary subgraph $H$ might contain different types of asymmetric structures. Our main idea is as follows. Given a sequence $S$, we construct a graph $G$ on the fly such that each element in the sequence corresponds to a source component of $H$ in $G$; an element in $S$ appears $k$ times if and only if a copy of $H$ is added in $G$. For the latter, we carefully add $k$ source components of $H$ to $G$ and add edges from these components to one center component (which is the rest part of $H$ after removing all the source components). Finally, we show that this construction preserves the distance to the properties and each query to $G$ can be answered by querying at most $1$ position in $S$. %

\subsection{Other Related work}
Ito, Khoury and Newman \cite{ito2020characterization} recently gave a characterization of monotone and hereditary properties that can be tested with constant query complexity and one-sided error in both bounded-degree bidirectional model and bounded-degree unidirectional model. For testing acyclicity in the bidirectional model, Bender and Ron \cite{bender2002testing} gave a lower bound of $\Omega(n^{1/3})$ queries for algorithms with two-sided error and a lower bound $\Omega(n^{1/2})$ queries for algorithms with one-sided error. The latter lower bound has been improved to $\tilde{\Omega}(n^{5/9})$ queries by Chen, Randolph, Servedio and Sun \cite{chen2020lower}. 

In the dense directed graph model (with different types of queries and notion of ``$\varepsilon$-far''), Alon and Shapira \cite{alon2004testing} gave an algorithm with constant query complexity for testing subgraph-freeness.

There exists a class of properties which can be tested with constant number of queries by the so-called \emph{proximity-oblivious testers} \cite{goldreich2011proximity}. Goldreich and Ron \cite{goldreich2016sample} showed that any property that can be tested by a proximity-oblivious tester that makes $q$ uniformly distributed queries with constant detection probability can be tested by a sample-based testers of sample complexity $O(n^{1-1/q})$, where a sample-based tester only samples elements independently from some distribution of the tested object. Building upon \cite{fischer2015trading,gur2021power}, Dall'Agnol, Tom and Lachish \cite{dall2021structural} recently showed that any property that is testable with $q$ queries admits a sample-based tester with sample complexity $n^{1-1/O(q^2\log^2 q)}$. Their algorithms are defined over a constant-size output alphabet, which is very different from the bounded degree (directed) graph model, in which a super constant alphabet is needed. 

Valiant developed a wishful thinking theorem in \cite{valiant2008testing}, telling that two distributions whose so-called \emph{$k$-based moments} have small gap are indistinguishable by $k$-Poissonized samples. This is a tool for establishing lower bounds of testing symmetric properties on distributions. On a very high level, both \cite{valiant2008testing} and our work are constructing far distributions with the same collision, while the details for the constructions differ significantly. For example, our proof is built upon Corollary 5.7 of \cite{raskhodnikova2009strong}, which requires to carefully construct two distributions that have proportional moments. In \cite{valiant2008testing}, it is required to construct two distributions whose $k$-based moments have small gap. It is unclear if two distributions with small gap between $k$-based moments have proportional moments, or vice versa. In addition, we are using very different properties of Vandermonde matrix from those used in \cite{valiant2008testing}. Although it is possible to obtain a lower bound for the $k$-occurrence-freeness testing problem by converting it into a distribution testing problem and subsequently employing Valiant's wishful thinking theorem, we believe that this approach yields a suboptimal bound compared to ours.

\section{A Lower Bound for Testing $k$-Occurrence-freeness}

In this section, we will prove the lower bound on the query complexity for testing $k$-occurrence-freeness, %
which is defined as follows. Given a sequence $A$ of $n$ integers such that each entry of $A$ is from $[n]:=\{1,\dots,n\}$ and each element $i\in [n]$ occurs at most $k$ times, the problem is to distinguish if $A$ is \emph{$k$-occurrence-free}, i.e., no element occurs in $k$ positions of $A$, or $A$ is \emph{$\varepsilon$-far from $k$-occurrence-free}, i.e., more than $\varepsilon n$ elements of $A$ needs to be changed to make it $k$-occurrence-free. We assume that the algorithm can query the element (or the value) of any position of the sequence in constant time. The goal is to solve the problem by making as few queries as possible. We will show the following result. 
\begin{theorem}\label{thm:kvalue-lowerbound}
	Any algorithm for testing $k$-occurrence-freeness with parameter $\varepsilon=\Omega_k(1)$ requires at least $\Omega(n^{1-1/k})$ queries, where $n$ is the length of the input sequence.
\end{theorem}

\subsection{Basic tools and notions} To prove the above theorem, we will make use of a lower bound by Raskhodnikova et al. \cite{raskhodnikova2009strong} for distinguishing two sequences satisfying some property. %
We first introduce two definitions.

\begin{definition}[\textbf{Frequency variable}]\label{def:freq var}
	Let $A$ be a sequence of integers. We define its \emph{frequency variable $X_A$} as follows. Choose a number uniformly at random from the set of distinct elements that occur in $A$ and then let $X_A$ denote its \emph{frequency}\footnote{We directly adopt the notion ``frequency'' from \cite{raskhodnikova2009strong}.}, i.e., the number of times it occurs. %
\end{definition}

Take the following sequence $S=\{1,2,1,3,2,1,4\}$ as an example. There are $4$ distinct elements (or values) in $S$: value $1$ occurs $3$ times, value $2$ occurs twice, value $3$ and $4$ each occurs once. Thus the frequency variable $X_S$ of $S$ satisfies that $\Pr[X_S=1]=0.5$, $\Pr[X_S=2]=0.25$, $\Pr[X_S=3]=0.25$.

\begin{definition}[\textbf{Proportional moments}]\label{def:k-prop moments}
	Two random variables $X_1$ and $X_2$ %
	are said to have \emph{$k-1$ proportional moments},
	if $ \frac{\E[X_2]}{\E[X_1]} = \frac{\E[X_2^2]}{\E[X_1^2]} = \cdots = \frac{\E[X_2^{k-1}]}{\E[X_1^{k-1}]}. $ {We say that two sequences have $k-1$ proportional moments if their frequency variables have $k-1$ proportional moments.}
\end{definition}

Let $P$ denote a property defined on sequence of integers such that it is invariant under any permutation of indices and values. \cite{raskhodnikova2009strong} has shown that any tester for $P$ that makes $t$ queries can be simulated by a \emph{Poisson}-$s$ algorithm that only looks at the histogram of the samples as its input, and $s=O(t)$. Relevant definitions are as follows.%

\begin{definition}[\textbf{\emph{Poisson}-$s$ algorithm}]\label{def:poisson-s alg}
	An algorithm is called a \emph{Poisson}-$s$ algorithm if the number of samples of the algorithm is determined by a Poisson distribution with the expectation $s$.
\end{definition}

\begin{definition}[\textbf{Histogram}]\label{def:histogram}
	Given a sequence $S$, the histogram $H$ of $S$ is a function defined as follows:$$H(i):=|\{ s\in S | \mbox{$s$ occurs exactly $i$ times in $S$}\}|$$
\end{definition}

In \cite{raskhodnikova2009strong}, Raskhodnikova et al. proved that if two sequences have $k-1$ proportional moments and $s=o(n^{1-\frac{1}{k}})$, then any \emph{Poisson}-$s$ algorithm can't distinguish their histograms.%
Formally, based on Lemma 5.3 and Corollary 5.7 in \cite{raskhodnikova2009strong}, we have the following Lemma.%

\begin{lemma}[\cite{raskhodnikova2009strong}]\label{lem:prop_moments}
	Let $X_A$ and $X_B$ be two random variables with $k-1$ proportional moments. And let $D_{X_A}$ and $D_{X_B}$ be two length-$n$ sequences of integers, whose frequency variables are $X_A$ and $X_B$, respectively. Let $P$ be a property of sequences that is invariant under permutations of indices and values, {and let $\varepsilon > 0$ be a constant}.
	
	\begin{enumerate}
		\item If $\mathcal{A'}$ is a tester for $P$ with $t$ queries, i.e., $\mathcal{A'}$ accepts the input sequence that satisfies $P$ with probability at least $\frac{2}{3}$; it rejects any sequence that is $\varepsilon$-far from satisfying $P$, with probability at least $\frac{2}{3}$.
		
		Then there must be a \emph{Poisson}-$s$ algorithm $\mathcal{A}$ that gets only the histogram of the samples, where $s=O(t)$, satisfiying the following: if the input sequence satisfies $P$, $\mathcal{A}$ accepts with probability at least $\frac{2}{3}-o(1)$; if the input sequence is $\varepsilon$-far from satisfying $P$, $\mathcal{A}$ rejects with probability at least $\frac{2}{3}-o(1)$.
		
		\item For any \emph{Poisson}-$s$ algorithm $\mathcal{A}$ with $s=o(n^{1-\frac{1}{k}})$, if $\mathcal{A}$ gets only access to the histogram of samples, then we have$$|\Pr[\mathcal{A}(D_{X_A})=\True]-\Pr[\mathcal{A}(D_{X_B})=\True]|=o(1).$$
	\end{enumerate}
\end{lemma}

Note that by the above Lemma, %
for a property $P$ that is invariant under permutation of indices and values, any tester for $P$ can be well simulated by a \emph{Poisson}-$s$ algorithm, which only accesses to the histogram of samples. Thus it suffices to only consider such \emph{Poisson}-$s$ algorithms. Furthermore, if there exist two instances of $P$ with proportional moments, then it is hard to distinguish these two instances, for any \emph{Poisson}-$s$ algorithm that only accesses to the histogram of samples.

\subsection{Proof of Theorem \ref{thm:kvalue-lowerbound}}
Now we give the proof of Theorem \ref{thm:kvalue-lowerbound}. %
We first note that $k$-occurrence-freeness is a property that is invariant under permutation of indices and values. Suppose that there exist two families of sequence instances, denoted by $\CA$ and $\CB$, respectively, such that 1) $\CA$ and $\CB$ have $k-1$ proportional moments; 2) sequences in $\CA$ are $k$-occurrence-free, and sequences in $\CB$ are far from $k$-occurrence-freeness. 
Now assume that there exist a tester $\mathcal{A'}$ for $k$-occurrence-freeness with $s=o(n^{1-\frac{1}{k}})$ queries. Then, according to Lemma \ref{lem:prop_moments}, there must be a \emph{Poisson}-$s$ algorithm $\mathcal{A}$ that gets only access to the histogram of samples. For such algorithm $\mathcal{A}$, we have
$$
|\Pr[\mathcal{A}(D_{X_A})=\True]-\Pr[\mathcal{A}(D_{X_B})=\True]| = (\frac{2}{3}-o(1))-(\frac{1}{3}+o(1)) \geq \frac{1}{6},
$$
which contradicts to the second part of Lemma \ref{lem:prop_moments} and thus implies the $\Omega(n^{1-\frac{1}{k}})$ lower bound. Therefore, to prove Theorem \ref{thm:kvalue-lowerbound}, it suffices to construct two families of sequences with the above desired properties. 

\begin{proof}[Proof of Theorem \ref{thm:kvalue-lowerbound}]\label{prove k-value}
	We first construct two classes, denoted by $\CA,\CB$, of length-$n$ sequences, such that for any sequences $A\in \CA$ and $B\in \CB$, it holds that 1) $A$ is $k$-occurrence-free and $B$ is $\varepsilon$-far from $k$-occurrence-free, and 2) the frequency variables $X_A, X_B$ of these two instances $A,B$ have $k-1$ proportional moments. 
	
	To do so, we first prove the claim.%
	\begin{claim}
		It holds that 
		$$\label{eq:eq mat}
		\begin{pmatrix}
			1      & 1       & \cdots & 1       \\
			1      & 2       & \cdots & k       \\
			1      & 2^2     & \cdots & k^2     \\
			\vdots & \vdots  & \ddots & \vdots  \\
			1      & 2^{k-1} & \cdots & k^{k-1} \\
		\end{pmatrix}
		\cdot
		\begin{pmatrix}
			(-1)^1\tbinom{k}{1} \\
			(-1)^2\tbinom{k}{2} \\
			(-1)^3\tbinom{k}{3} \\
			\vdots              \\                (-1)^k\tbinom{k}{k} \\
		\end{pmatrix} 
		=
		\begin{pmatrix}
			-1     \\
			0      \\
			0      \\
			\vdots \\
			0
		\end{pmatrix}.
		$$
	\end{claim}
	\begin{proof}
		We define a sequence of helper functions $f_j(x)$ to prove (\ref{eq:eq mat}). %
		$$\label{def:fs}
		f_j(x) =
		\begin{cases}
			(1+x)^k,          & \mbox{$j=0$}                \\
			x\cdot f^{'}_{j-1}(x), & \mbox{$j=1,2,\cdots,k-1$}
		\end{cases}
		$$
		To prove the claim, we note that it suffices to show the following: 
		\begin{eqnarray} f_0(-1)&=&1+\sum_{i=1}^{k}(-1)^\cdot \tbinom{k}{j}=0, \label{equation-the-first}\\
			f_j(-1)&=&\sum_{i=1}^ki^j\cdot (-1)^i\cdot \tbinom{k}{i}=0, \textrm{ for any $j=1,\dots,k-1$.} \label{equation-the-second}
		\end{eqnarray}

		Note that if the above are true, then each line of Equations (\ref{eq:eq mat}) holds, which finishes the proof of the claim. In the following, we prove Equations (\ref{equation-the-first}) and (\ref{equation-the-second}).

		Let us first consider the binomial expansion of $f_0(x)$. We have that
		$$\label{eq:f0_expansion}
		f_0(x)={(1+x)}^k=\ \sum_{i=0}^{k}{x^i\cdot \tbinom{k}{i}=1+}\sum_{i=1}^{k}{x^i\cdot \tbinom{k}{i}}.
		$$
		Thus, $f_0\left(-1\right)=\left(1-1\right)^k=\ 1+\sum_{i=1}^{k}{{(-1)}^i\cdot\tbinom{k}{i}}=0$. That is, Equation (\ref{equation-the-first}) holds.

		To prove Equation (\ref{equation-the-second}), we show that for any $1\leq j\leq k-1$, it holds that
		\begin{enumerate}
			\item[(a)]\label{fj_first} $f_j(x)=\sum_{i=1}^{k}i^j\cdot x^i\cdot \tbinom{k}{i}$, 
			\item[(b)]\label{fj_second} $f_j(x)=\sum_{i=1}^j a_i\cdot x^i\cdot (1+x)^{k-i}$, for some numbers $a_1,\dots,a_j\geq 0$. 
		\end{enumerate}
		Note that by the above two items, we have that $f_j(-1)=0=\sum_{i=1}^{k}i^j\cdot x^i\cdot\tbinom{k}{i}$, for each $j=1,\dots,k-1$, which finishes the proof of Equation (\ref{equation-the-second}) (and the claim). 
		
		In the following, we prove the above two items (a) and (b) by induction. Consider the case $j=1$. By definition of function $f_j(x)$ given by ({\ref{def:fs}}) and the expansion (\ref{eq:f0_expansion}), it holds that
		$$
		{f'_0}\left(x\right)=k\cdot {(1+x)}^{k-1}=\sum_{i=1}^{k}i\cdot x^{i-1}\cdot\tbinom{k}{i},$$
		which implies that
		$$\label{eq:l2}
		f_1(x)=x\cdot  f^{'}_{0}(x)=x\cdot k\cdot (1+x)^{k-1}=\sum_{i=1}^{k}i\cdot x^i\cdot \tbinom{k}{i}
		$$
		Now we assume that the items (a) and (b) hold for $j\leq k-2$, and we prove it for $j+1$. 
		For item (a), since $f_j(x)=\sum_{i=1}^{k}i^j\cdot x^i\cdot \tbinom{k}{i}$, we have that $f'_j(x)=\sum_{i=1}^{k}i^{j+1}\cdot x^{i-1}\cdot \tbinom{k}{i}$. Thus, $$f_{j+1}(x)= \sum_{i=1}^{k}i^{j+1}\cdot x^{i}\cdot \tbinom{k}{i}$$ 
		by Definition (\ref{def:fs}).

		For item (b), since $f_j(x)=\sum_{i=1}^j a_i\cdot x^i\cdot (1+x)^{k-i}$, for some numbers $a_1,\dots,a_j\ge 0$, it holds that $$f_j'(x)=\sum_{i=1}^j(a_i\cdot i\cdot x^{i-1}\cdot (1+x)^{k-i}+a_i\cdot x^i\cdot (k-i)\cdot (1+x)^{k-i-1}).$$ Thus, by Definition (\ref{def:fs}),  
		$$f_{j+1}(x)=\sum_{i=1}^j(a_i\cdot i\cdot x^{i}\cdot (1+x)^{k-i}+a_i\cdot x^{i+1}\cdot (k-i)\cdot (1+x)^{k-i-1})=\sum_{i=1}^{j+1}a'_i \cdot x^{i}\cdot (1+x)^{k-i},$$ 
		for some numbers $a_1',\cdots,a_{j+1}'\geq 0$. 
		
		Therefore, both items (a) and (b) hold and this finishes the proof the claim.
	\end{proof}

	Now we define two distributions $\p,\q$ over $[k]$ as follows. 
	\begin{enumerate}
		\item if $k$ is even, define
		\begin{align*}
			\p_i=  \
			\begin{cases}
				0,                               & \mbox{if $i$ is even} \\
				\frac{1}{2^{k-1}}\cdot \tbinom{k}{i}, & \mbox{if $i$ is odd}
			\end{cases} 
			\qquad   \q_i= 
			\begin{cases}
				\frac{1}{2^{k-1}-1}\cdot \tbinom{k}{i}, & \mbox{if $i$ is even} \\
				0,                                 & \mbox{if $i$ is odd}
			\end{cases} 
		\end{align*}
		
		\item if $k$ is odd, define
		\begin{align*}
			\p_i= 
			\begin{cases}
				\frac{1}{2^{k-1}-1}\cdot \tbinom{k}{i}, & \mbox{if $i$ is even} \\
				0,                                 & \mbox{if $i$ is odd}
			\end{cases} 
			\qquad \q_i= 
			\begin{cases}
				0,                               & \mbox{if $i$ is even} \\
				\frac{1}{2^{k-1}}\cdot \tbinom{k}{i}, & \mbox{if $i$ is odd}
			\end{cases} 
		\end{align*}
	\end{enumerate}

	Now we show the following Lemma.
	\begin{lemma}\label{lemma:definitionsofpq}
		Let $\p,\q$ be defined as above. There exists $d>0$ such that  
		$$\label{eq:solution}
		\begin{pmatrix}
			\q_1    \\
			\q_2    \\
			\q_3    \\
			\vdots \\
			\q_k    \\
		\end{pmatrix}
		=
		d \cdot
		\begin{pmatrix}
			\p_1    \\
			\p_2    \\
			\p_3    \\
			\vdots \\
			\p_k    \\
		\end{pmatrix}
		+
		(d-1) \cdot
		\begin{pmatrix}
			(-1)^1\tbinom{k}{1} \\
			(-1)^2\tbinom{k}{2} \\
			(-1)^3\tbinom{k}{3} \\
			\vdots              \\
			(-1)^k\tbinom{k}{k} \\
		\end{pmatrix}
		$$
	\end{lemma}
	
	\begin{proof}
		For the case that $k$ is even, we let $d=1+\frac{1}{2^{k-1}-1}$.  
		
		First note that $\p_k=0$ and $\q_k=(d-1)\cdot \tbinom{k}{k}$. Thus, the last equation holds. For even $i\in \{2,4,\ldots,k \}$, $\p_i=0$ and $\q_i=d\cdot \p_i+(d-1)\cdot \tbinom{k}{i}$. For odd $i\in \{1,3,\ldots,k-1 \}$, $\p_i=\frac{d-1}{d}\cdot \tbinom{k}{i}$ and $\q_i=d\cdot \p_i + (d-1)\cdot (-1)\cdot \tbinom{k}{i}=0$. Thus, Equation (\ref{eq:solution}) holds.
		
		For the case that $k$ is odd, we let $d=1-\frac{1}{2^{k-1}}$. 
		
		Note that $\p_k=0$ and $\q_k=(1-d)\cdot \tbinom{k}{k}$. Thus, the last equation holds. For odd $i\in \{1,3,\ldots,k \}$, $\p_i=0$ and $\q_i=d\cdot \p_i + (1-d)\cdot \tbinom{k}{i}$.
		For even $i\in \{2,4,\ldots,k-1 \}$, $\p_i=\frac{1-d}{d}\cdot\tbinom{k}{i}$ and $\q_i=d\cdot \p_i+(d-1)\cdot \tbinom{k}{i}=0$. Thus, Equation (\ref{eq:solution}) holds.
	\end{proof}

	\begin{lemma}
		Let $k$ be any integer with $k\geq 2$. Let $\p,\q$ be distributions over $[k]$  defined as above. It holds that 
		\begin{enumerate}
			\item $\p_k=0$ and $\q_k \geq \frac{1}{2^k}$;
			\item for any two random variables $X_A$ and $X_B$ with distributions $\p$ and $\q$, respectively, it holds that $X_A$ and $X_B$ have $k-1$ proportional moments.
		\end{enumerate}
	\end{lemma}
	
	\begin{proof}
		The first item follows from the definitions of $\p$ and $\q$. %
		
		Now prove the second item. Let $d>0$ be the number from Lemma \ref{lemma:definitionsofpq}. We will show that 
		$$
		\frac{\E[X_B]}{\E[X_A]} = \frac{\E[X_B^2]}{\E[X_A^2]} = \cdots = \frac{\E[X_B^{k-1}]}{\E[X_A^{k-1}]} = d, %
		$$
		or equivalently, 
		$$\label{eq:proportional matrix}
		\begin{pmatrix}
			1            \\
			\E[X_B]       \\
			\E[X_B^2]     \\
			\vdots       \\
			\E[X_B^{k-1}] \\
		\end{pmatrix}
		= d \cdot
		\begin{pmatrix}
			1/d          \\
			\E[X_A]       \\
			\E[X_A^2]     \\
			\vdots       \\
			\E[X_A^{k-1}] \\
		\end{pmatrix}.
		$$
		
		By the definition $X_A$, it holds that for any $0\leq i\leq k-1$, $\E[X_A^{i}]=\sum_{j=1}^kp_j\cdot j^i$. That is, 
		$$\label{eq:moments of X_A}
		\begin{pmatrix}
			1            \\
			\E[X_A]       \\
			\E[X_A^2]     \\
			\vdots       \\
			\E[X_A^{k-1}] \\
		\end{pmatrix}
		=
		\begin{pmatrix}
			1      & 1       & \cdots & 1       \\
			1      & 2       & \cdots & k       \\
			1      & 2^2     & \cdots & k^2     \\
			\vdots & \vdots  & \ddots & \vdots  \\
			1      & 2^{k-1} & \cdots & k^{k-1} \\
		\end{pmatrix}
		\begin{pmatrix}
			\p_1    \\
			\p_2    \\
			\p_3    \\
			\vdots \\
			\p_k    \\
		\end{pmatrix}
		$$
		Similarly, it holds that
		$$\label{eq:moments of X_B}
		\begin{pmatrix}
			1            \\
			\E[X_B]       \\
			\E[X_B^2]     \\
			\vdots       \\
			\E[X_B^{k-1}] \\
		\end{pmatrix}
		=
		\begin{pmatrix}
			1      & 1       & \cdots & 1       \\
			1      & 2       & \cdots & k       \\
			1      & 2^2     & \cdots & k^2     \\
			\vdots & \vdots  & \ddots & \vdots  \\
			1      & 2^{k-1} & \cdots & k^{k-1} \\
		\end{pmatrix}
		\begin{pmatrix}
			\q_1    \\
			\q_2    \\
			\q_3    \\
			\vdots \\
			\q_k    \\
		\end{pmatrix}
		$$

		By Equations (\ref{eq:moments of X_A}) and (\ref{eq:moments of X_B}), we know that to prove Equation (\ref{eq:proportional matrix}), it suffices to show that
		$$\label{eq:core matrix}
		\begin{pmatrix}
			1      & 1       & \cdots & 1       \\
			1      & 2       & \cdots & k       \\
			1      & 2^2     & \cdots & k^2     \\
			\vdots & \vdots  & \ddots & \vdots  \\
			1      & 2^{k-1} & \cdots & k^{k-1} \\
		\end{pmatrix}
		\begin{pmatrix}
			\q_1    \\
			\q_2    \\
			\q_3    \\
			\vdots \\
			\q_k    \\
		\end{pmatrix}
		=
		d \cdot 
		\begin{pmatrix}
			1/d    & 1/d     & \cdots & 1/d     \\
			1      & 2       & \cdots & k       \\
			1      & 2^2     & \cdots & k^2     \\
			\vdots & \vdots  & \ddots & \vdots  \\
			1      & 2^{k-1} & \cdots & k^{k-1} \\
		\end{pmatrix}
		\begin{pmatrix}
			\p_1    \\
			\p_2    \\
			\p_3    \\
			\vdots \\
			\p_k    \\
		\end{pmatrix}.
		$$
		
		Recall that by Lemma \ref{lemma:definitionsofpq}, it holds that
		$$%
		\begin{pmatrix}
			\q_1    \\
			\q_2    \\
			\q_3    \\
			\vdots \\
			\q_k    \\
		\end{pmatrix}
		=
		d \cdot
		\begin{pmatrix}
			\p_1    \\
			\p_2    \\
			\p_3    \\
			\vdots \\
			\p_k    \\
		\end{pmatrix}
		+
		(d-1) \cdot
		\begin{pmatrix}
			(-1)^1\tbinom{k}{1} \\
			(-1)^2\tbinom{k}{2} \\
			(-1)^3\tbinom{k}{3} \\
			\vdots              \\
			(-1)^k\tbinom{k}{k} \\
		\end{pmatrix}
		$$

		Substituting $\q_i$ from the above equation to the left hand side of equation (\ref{eq:core matrix}) gives us that 
		\begin{align*}
			& d \cdot
			\begin{pmatrix}
				1      & 1       & \cdots & 1       \\
				1      & 2       & \cdots & k       \\
				1      & 2^2     & \cdots & k^2     \\
				\vdots & \vdots  & \ddots & \vdots  \\
				1      & 2^{k-1} & \cdots & k^{k-1} \\
			\end{pmatrix}
			\cdot
			\begin{pmatrix}
				\p_1    \\
				\p_2    \\
				\p_3    \\
				\vdots \\
				\p_k    \\
			\end{pmatrix}
			+
			(d-1) \cdot
			\begin{pmatrix}
				1      & 1       & \cdots & 1       \\
				1      & 2       & \cdots & k       \\
				1      & 2^2     & \cdots & k^2     \\
				\vdots & \vdots  & \ddots & \vdots  \\
				1      & 2^{k-1} & \cdots & k^{k-1} \\
			\end{pmatrix}
			\cdot
			\begin{pmatrix}
				(-1)^1\tbinom{k}{1} \\
				(-1)^2\tbinom{k}{2} \\
				(-1)^3\tbinom{k}{3} \\
				\vdots              \\
				(-1)^k\tbinom{k}{i} \\
			\end{pmatrix}  \\
			= & d \cdot
			\begin{pmatrix}
				1            \\
				\E[X_A]       \\
				\E[X_A^2]     \\
				\vdots       \\
				\E[X_A^{k-1}] \\
			\end{pmatrix}
			+
			(d-1) \cdot
			\begin{pmatrix}
				1      & 1       & \cdots & 1       \\
				1      & 2       & \cdots & k       \\
				1      & 2^2     & \cdots & k^2     \\
				\vdots & \vdots  & \ddots & \vdots  \\
				1      & 2^{k-1} & \cdots & k^{k-1} \\
			\end{pmatrix}
			\cdot
			\begin{pmatrix}
				(-1)^1\tbinom{k}{1} \\
				(-1)^2\tbinom{k}{2} \\
				(-1)^3\tbinom{k}{3} \\
				\vdots              \\
				(-1)^k\tbinom{k}{i} \\
			\end{pmatrix}= d \cdot
			\begin{pmatrix}
				1/d            \\
				\E[X_A]       \\
				\E[X_A^2]     \\
				\vdots       \\
				\E[X_A^{k-1}] \\
			\end{pmatrix},\\
		\end{align*}
		where the last equation follows from Claim \ref{eq:eq mat}.
		
		On the other hand, by Equation (\ref{eq:moments of X_A}), we know that the right hand side of (\ref{eq:core matrix}) is, 
		$$
		d \cdot
		\begin{pmatrix}
			1/d    & 1/d     & \cdots & 1/d     \\
			1      & 2       & \cdots & k       \\
			1      & 2^2     & \cdots & k^2     \\
			\vdots & \vdots  & \ddots & \vdots  \\
			1      & 2^{k-1} & \cdots & k^{k-1} \\
		\end{pmatrix}
		\begin{pmatrix}
			\p_1    \\
			\p_2    \\
			\p_3    \\
			\vdots \\
			\p_k    \\
		\end{pmatrix}
		= d \cdot
		\begin{pmatrix}
			1/d          \\
			\E[X_A]       \\
			\E[X_A^2]     \\
			\vdots       \\
			\E[X_A^{k-1}] \\
		\end{pmatrix}.
		$$
		Therefore, Equation  (\ref{eq:core matrix}) holds and thus $X_A$ and $X_B$ have $k-1$ proportional moments. This finishes the proof of the Lemma.
	\end{proof}

	Now we construct class $\CA$ as follows: $\CA$ is a class of sequences, and the frequency variable $X_A$ of every sequence $A$ is $\Pr[X_A = i] = p_i$. That is, for every sequence $A$, the fraction of elements that occur $i$ times is exactly $p_i$. We can construct $\CB$ analogously by substituting $p_i$ with $q_i$.
	
	By construction, sequence $A$ is $k$-occurrence-free. Consider the sequence $B$. Suppose that there are $l$ distinct values in $B$, then at least $q_k \cdot l$ values occur $k$ times in $B$, which means that $B$ is at least $\frac{q_k \cdot l}{n}$-far from $k$-occurrence-free. As every value in $B$ occurs in at most $k$ positions, there are at least $\frac{n}{k}$ distinct values, i.e., $l \geq \frac{n}{k}$. Thus, $B$ is at least $\frac{q_k}{k}$-far from $k$-occurrence-free. {According to previous analysis, $A$ and $B$ have $k-1$ proportional moments.} The theorem then follows from Lemma \ref{lem:prop_moments}. 
\end{proof}

\section{A Lower Bound for Testing Subgraph-Freeness} 
In this section, we give the proof of the lower bound on the query complexity for testing subgraph-freeness, i.e., prove Theorem \ref{thm:subgraph-lowerbound}. 

\begin{proof}[Proof of Theorem \ref{thm:subgraph-lowerbound}]
	We give a reduction from the problem of testing $k$-occurrence of a sequence to the problem of testing $H$-freeness in the unidirectional model. That is, given an instance of the former problem, i.e., a length-$n$ sequence $S$ such that each element is promised to occur at most $k$ times, we will construct an instance of the $H$-freeness testing problem, i.e., a directed graph $G$ with $n'=\Theta(n)$ vertices and bounded degree. Then we show that this construction preserves the distances of the properties and any algorithm $\mathcal{A}'$ for testing $H$-freeness in the unidirectional model can be invoked on $G$ to test if $S$ is $k$-occurrence-freeness. In particular, if $\mathcal{A}'$ has query complexity $o(n'^{1-\frac{1}{k}})$, then this implies an algorithm for testing $k$-occurrence-freeness with query complexity $o(n^{1-\frac{1}{k}})$, contradicting to Theorem \ref{thm:kvalue-lowerbound}.

	\vspace{1em}
	\noindent\textbf{Preprocessing the subgraph $H$.} Since $H$ has $k$ source components, we denote these components by $\{C_1,\cdots,C_k\}$. Note that each $C_i$ is a subgraph of $H$. We use $N_{\comp}$ to denote the maximum number of vertices in $\{C_1,C_2,\cdots,C_k\}$, i.e., $N_{\comp}=\mathop{\max}_{i=1,\cdots,k}|V(C_i)|$ where $V(C)$ denotes the vertex set of the graph $C$. 
	We use $C_0$ to denote the subgraph induced by the remainder of vertices in $V(H)$ and we call $C_0$ the center component of $H$. Let $N_{\ctr}=|V(C_0)|=|V(H)|-\sum_{i=1}^{k}|V(C_i)|$. Note that since $C_1,\cdots,C_k$ are source components, by definition, no edge exists between different such components. All the edges leaving $C_i$ (for $i=1,\cdots,k$) are entering $C_0$. We can first decompose $H$ into source components and the center component (e.g., by using Tarjan's algorithm \cite{tarjan1972depth}), index them, and identify all the edges crossing different components in constant time (as the size of $H$ is constant).

	We illustrate such a decomposition of a subgraph $\tilde{H}$ in Figure \ref{fig:H}. Note that $\tilde{H}$ has $3$ source components and $1$ center component (see Figure \ref{fig:component}). It can be partitioned into four parts such  $V(\tilde{C_0})=\{v_2,v_7\}$, $V(\tilde{C_1})=\{v_1\}$, $V(\tilde{C_2})=\{v_3,v_4,v_5\}$, $V(\tilde{C_3})=\{v_6\}$  as follows. In this example, $N_{\comp}=3$, $N_{\ctr}=2$. %

	In the construction of the graph $G$, we will treat each component $C_i$, $1\leq i\leq k$, as a subgraph with $N_{\comp}$ vertices. That is, for each such $C_i$, we add %
	$(N_{\comp}-|V(C_i)|)$ isolated vertices to $C_i$ to obtain a new component $C_i'$ so that $|V(C_i')|=N_{\comp}$. We can reassemble these new components $\{C_1',C_2',\cdots,C_k'\}$ with $C_0$ to obtain a graph $H'$.%
	
	\begin{figure*}[htb]
		\centering
		\begin{tikzpicture}
			\tikzstyle{edge}=[line width=1.3, -latex]
			\tikzstyle{vertex}=[draw,circle,fill=black,inner sep=0pt, minimum width=9pt]
			\def \dist {1.5}
			\node[vertex] (1) at (0,1*\dist) {};	
			\node[vertex] (2) at (0,0) {};
			\node[vertex] (3) at (-1*\dist,-1*\dist) {};
			\node[vertex] (4) at (-2*\dist,-2*\dist) {};
			\node[vertex] (5) at (0,-2*\dist) {};
			\node[vertex] (6) at (2*\dist,-2*\dist) {};
			\node[vertex] (7) at (1*\dist,-1*\dist) {};
			\draw[edge] (1)--(2);
			\draw[edge] (4)--(3);
			\draw[edge] (3)--(2);
			\draw[edge] (6)--(7);
			\draw[edge] (7)--(2);
			\draw[edge] (3)--(5);
			\draw[edge] (5)--(4);
			
			\node at (1)[above=4pt]{$v_{1}$};
			\foreach \x in {2,...,7}\node at (\x)[below=4pt]{$v_{\x}$};
		\end{tikzpicture}
		\caption{A subgraph $\tilde{H}$} %
		\label{fig:H}
	\end{figure*}
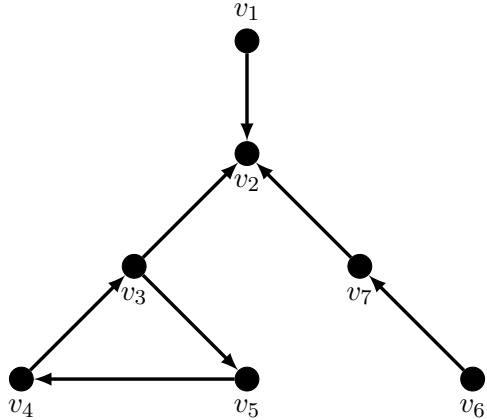
	
	Now we index each vertex of $H'$ by some integer in $\{1,\cdots, N_{\ctr}+k\cdot N_{\comp}\}$ as follows. The index set of $V(C_0')$ is $[1,N_{\ctr}]$, and the index set of $V(C_i')$ is $[N_{\ctr}+(i-1)\cdot N_{\comp} +1, N_{\ctr}+i\cdot N_{\comp}]$, for each $1\leq i\leq k$. Furthermore, for each component $C_i'$ with $0\leq i\leq k$, we sequentially index the vertices using the corresponding index set according to the lexicographical ordering of the vertices in the aforementioned component decomposition.

	\begin{figure*}[htb]
		\centering
		\begin{subfigure}{0.4\linewidth}
			\centering
			\begin{tikzpicture}
				\tikzstyle{edge}=[line width=1.3, -latex]
				\tikzstyle{vertex}=[draw,circle,fill=black,inner sep=0pt, minimum width=9pt]
				\def \dist {1.5}
				\node[vertex] (1) at (0,1*\dist) {};
				\node at (1)[above=4pt]{$v_{1}$};
			\end{tikzpicture}
			\subcaption{source component $\tilde{C_1}$}
		\end{subfigure}%
		\begin{subfigure}{0.4\linewidth}
			\centering
			\begin{tikzpicture}
				\tikzstyle{edge}=[line width=1.3, -latex]
				\tikzstyle{vertex}=[draw,circle,fill=black,inner sep=0pt, minimum width=9pt]
				\def \dist {1.5}
				\node[vertex] (3) at (-1*\dist,-1*\dist) {};
				\node[vertex] (4) at (-2*\dist,-2*\dist) {};
				\node[vertex] (5) at (0,-2*\dist) {};
				\draw[edge] (4)--(3);
				\draw[edge] (3)--(5);
				\draw[edge] (5)--(4);
				
				\foreach \x in {3,4,5}\node at (\x)[below=4pt]{$v_{\x}$};
			\end{tikzpicture}
			\subcaption{source component $\tilde{C_2}$}
		\end{subfigure}%
		
		\begin{subfigure}{0.4\linewidth}
			\centering
			\begin{tikzpicture}
				\tikzstyle{edge}=[line width=1.3, -latex]
				\tikzstyle{vertex}=[draw,circle,fill=black,inner sep=0pt, minimum width=9pt]
				\def \dist {1.5}
				\node[vertex] (6) at (2*\dist,-2*\dist) {};
				
				\foreach \x in {6}\node at (\x)[below=4pt]{$v_{\x}$};
			\end{tikzpicture}
			\subcaption{source component $\tilde{C_3}$}
		\end{subfigure}%
		\begin{subfigure}{0.4\linewidth}
			\centering
			\begin{tikzpicture}
				\tikzstyle{edge}=[line width=1.3, -latex]
				\tikzstyle{vertex}=[draw,circle,fill=black,inner sep=0pt, minimum width=9pt]
				\def \dist {1.5}
				\node[vertex] (2) at (0,0) {};
				\node[vertex] (7) at (1*\dist,-1*\dist) {};
				\draw[edge] (7)--(2);
				
				\foreach \x in {2,7}\node at (\x)[below=4pt]{$v_{\x}$};
			\end{tikzpicture}
			\subcaption{center component $\tilde{C_0}$}
		\end{subfigure}%
		\caption{Decomposing $\tilde{H}$ into $3$ source components and $1$ center component}
		\label{fig:component}
	\end{figure*}
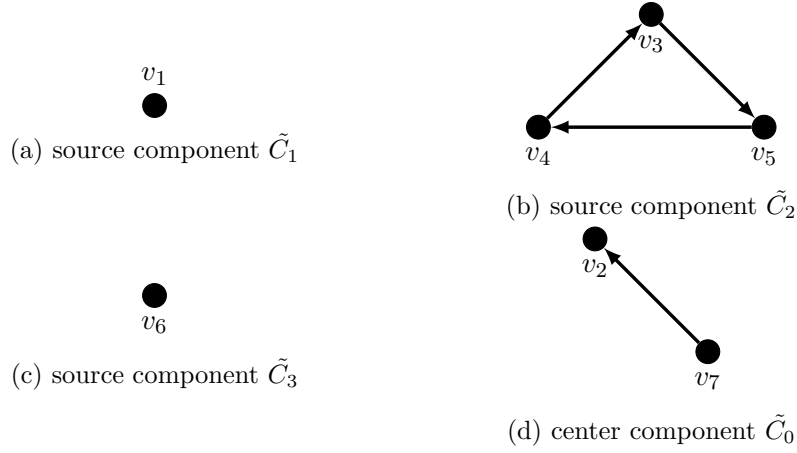
	
	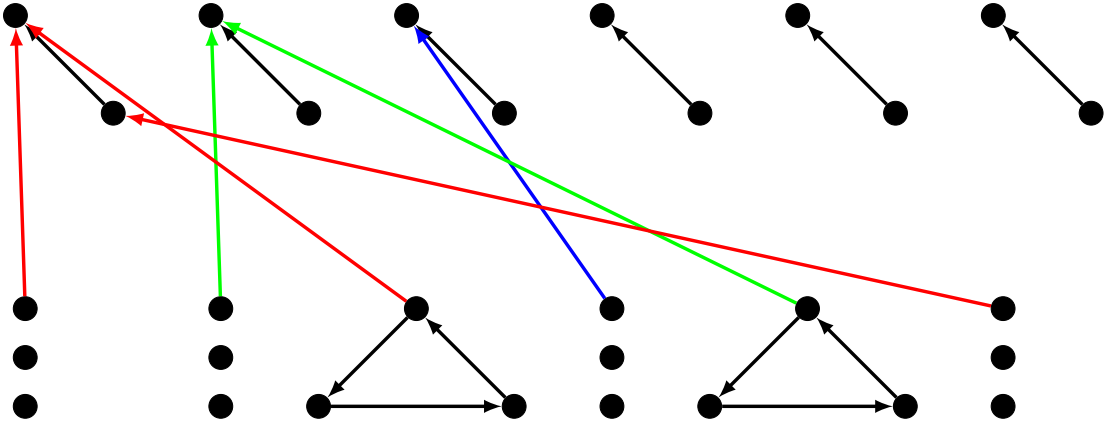
\begin{figure*}[htb]
		\begin{tikzpicture}
			\tikzstyle{edge}=[line width=1.3, -latex]
			\tikzstyle{vertex}=[draw,circle,fill=black,inner sep=0pt, minimum width=9pt]
			\def \dist {1.3}

			\node[vertex] (1) at (-12*\dist,0*\dist) {};
			\node[vertex] (2) at (-11*\dist,-1*\dist) {};
			\draw[edge] (2)--(1);
			\node[vertex] (3) at (-10*\dist,0*\dist) {};
			\node[vertex] (4) at (-9*\dist,-1*\dist) {};
			\draw[edge] (4)--(3);
			\node[vertex] (5) at (-8*\dist,0*\dist) {};
			\node[vertex] (6) at (-7*\dist,-1*\dist) {};
			\draw[edge] (6)--(5);
			\node[vertex] (7) at (-6*\dist,0*\dist) {};
			\node[vertex] (8) at (-5*\dist,-1*\dist) {};
			\draw[edge] (8)--(7);
			\node[vertex] (9) at (-4*\dist,0*\dist) {};
			\node[vertex] (10) at (-3*\dist,-1*\dist) {};
			\draw[edge] (10)--(9);
			\node[vertex] (11) at (-2*\dist,0*\dist) {};
			\node[vertex] (12) at (-1*\dist,-1*\dist) {};
			\draw[edge] (12)--(11);

			\node[vertex] (13) at (-11.9*\dist,-3*\dist) {};
			
			\node[vertex] (14) at (-9.9*\dist,-3*\dist) {};
			
			\node[vertex] (15) at (-7.9*\dist,-3*\dist) {};
			\node[vertex] (16) at (-6.9*\dist,-4*\dist) {};
			\node[vertex] (17) at (-8.9*\dist,-4*\dist) {};
			\draw[edge] (16)--(15);
			\draw[edge] (15)--(17);
			\draw[edge] (17)--(16);
			
			\node[vertex] (18) at (-5.9*\dist,-3*\dist) {};
			
			\node[vertex] (19) at (-3.9*\dist,-3*\dist) {};
			\node[vertex] (20) at (-2.9*\dist,-4*\dist) {};
			\node[vertex] (21) at (-4.9*\dist,-4*\dist) {};
			\draw[edge] (20)--(19);
			\draw[edge] (19)--(21);
			\draw[edge] (21)--(20);
			
			\node[vertex] (22) at (-1.9*\dist,-3*\dist) {};

			\node[vertex] (23) at (-11.9*\dist,-3.5*\dist) {};
			\node[vertex] (23) at (-11.9*\dist,-4*\dist) {};

			\node[vertex] (23) at (-9.9*\dist,-3.5*\dist) {};
			\node[vertex] (23) at (-9.9*\dist,-4*\dist) {};

			\node[vertex] (23) at (-5.9*\dist,-3.5*\dist) {};
			\node[vertex] (23) at (-5.9*\dist,-4*\dist) {};

			\node[vertex] (23) at (-1.9*\dist,-3.5*\dist) {};
			\node[vertex] (23) at (-1.9*\dist,-4*\dist) {};

			\draw[edge][red] (13)--(1);
			\draw[edge][green] (14)--(3);
			\draw[edge][red] (15)--(1);
			\draw[edge][blue] (18)--(5);
			\draw[edge][green] (19)--(3);
			\draw[edge][red] (22)--(2);
			
		\end{tikzpicture}
		\caption{Constructing $\tilde{G}$ from the sequence $\tilde{S}$ and the decomposition of $\tilde{H}$}
		\label{fig:G}
	\end{figure*}

	\vspace{1em}
	Now we describe the reduction. %
	Given a length-$n$ sequence $S$, the directed graph $G=(V,E)$ can be constructed as follows. We first add $n$ disjoint copies of the subgraph $C_0$ to $G$. Then we will add $n$ copies of source components and add some edges from source components to some copy of $C_0$ constructed before. That is,  each element in the sequence corresponds to a source component. %
	Note that there are no edges between different copies of source components. The offline construction is formally described as follows.
	
	\vspace{1em}
	\noindent\textbf{Vertex set and vertex indices.} We index vertices in $G$ from $1$ to $n\cdot(N_{\ctr}+N_{\comp})$. The vertex set is decomposed into two parts: the \emph{center part} and the \emph{source part}. More precisely, the source part contains $n$ potential source components with vertex indices  from $1$ to $n\cdot N_{\comp}$, and the center part contains $n$ disjoint copies of the center component $C_0$ with vertex indices from $n\cdot N_{\comp}+1$ to $n\cdot (N_{\comp}+N_{\ctr})$. Furthermore, the vertices in the $i$-th copy of the source component are indexed from $(i-1)\cdot N_{\comp}+1$ to $i\cdot N_{\comp}$, while the vertices of the $j$-th copy of the center component are indexed from $n\cdot N_{\comp} + (j-1) \cdot N_{\ctr} + 1$ to $n\cdot N_{\comp} + j \cdot N_{\ctr} $. 
	
	\vspace{1em}
	\noindent\textbf{Adding components and edges.}  Add $n$ disjoint copies of $C_0$ to $G$. Initialize a size-$n$ array $T$ such that $T_a=0$ for each $1\leq a\leq n$. For each $a = 1,2,\cdots,n$: 
	\begin{enumerate}
		\item let $b$ be the value (or element) of $S$ at position $a$
		\item If $b$ is a new value that algorithm sees for the first time,  define an array $R_b=\{1,2,\cdots,k\}$.  %
		\item %
		Uniformly sample a number $t$ from $R_b$. %
		Add an copy of $C'_t$. Ignoring isolated vertices in $C_t'$, add edges between this copy of $C_t'$ and the $b$-th copy of $C_0$ in the same way as the connections between their counterparts in the subgraph $H$. %
		Delete $t$ from $R_b$. Set $T_a=t$, i.e., the $a$-th position of $S$ is mapped to a source component $C_t'$.
	\end{enumerate}
	
	Note that by construction, the graph $G$ is $d$-bounded, and its maximum (in- or out-) degree the same as the maximum (in- or out-) degree of $H$.

	We give an illustration of the above construction in Figure \ref{fig:G}. Given a sequence $\tilde{S}=\{1,2,1,3,2,1\}$, and a subgraph $\tilde{H}$ as shown in Figure \ref{fig:H}. The graph $\tilde{G}$ from the above reduction is shown Figure \ref{fig:G}. In this figure, edges of the same color  correspond to positions of the same value (or element) in $\tilde{S}$. For example, the $3$ red edges correspond to the $3$ occurrences of value $1$. Together with the corresponding source and center components, these red edges form an copy of $\tilde{H}$.

	\vspace{1em}
	\noindent\textbf{Construction on the fly.} We show that the above construction of $G$ can be done on the fly and each query to $G$ can be answered by querying at most $1$ position in $S$. More precisely, let $\mathcal{A}'$ be an algorithm for testing $H$-freeness. When $\mathcal{A}'$ queries the $i$-th outgoing neighbor of a vertex $v$, we consider the following cases. 
	
	If $v>n\cdot N_{\comp}$, then $v$ belongs to a copy of $C_0$, then we do not need to query sequence $S$, and  we can simply locate the vertex $v'=(v-n\cdot N_{\comp}) \bmod{N_{\ctr}}$ in $C_0$. And by our index in $H'$, we know the corresponding vertex index in $H'$ is also $v'$. Then we can check the $i$-th neighbor of $v'$ in $H'$, denoted by $v''$. Thus we just return $v-v'+v''$. %
	
	If $1\leq v \leq n\cdot N_{\comp}$, then $v$ belongs to a copy of some source component. Calculate $a = \lceil{v/N_{\comp}}\rceil$ and query the $a$-th position of $S$. Let $b$ denote the query answer. If $T_a=0$, which means that this element is queried for the first time, uniformly sample a type $t$ from the rest of types $R_b$ for value $b$, and update $T_a=t$; otherwise simply set $t=T_a$. Note that $R$ and $T$ are maintained as described in the construction. Then calculate $v'=v \bmod{N_{\comp}}$. Now we know that the queried vertex $v$ corresponds to the $v'$-th vertex in a $C_t'$ component, which is adjacent to the $b$-th copy of $C_0$. We can look up vertex $N_{\ctr} + (t-1)\cdot N_{\comp} + v'$ in $H'$, which is isomorphic to vertex $v$ in $G$. We use $v''$ to denote the $i$-th neighbor of $N_{\ctr} + (t-1)\cdot N_{\comp} + v'$ in $H'$. If $v''$ belongs to the $C_t'$ part in $H$, we just return $v-v'+v''$. Otherwise, if $v''$ belongs to a $C_0$ part, we return $n\cdot N_{\comp} + (b-1)\cdot N_{\ctr} + v''$.

	Thus, any query for a vertex $v$ with $v>n\cdot N_{\comp}$ can be answered without querying $S$; query for a vertex $v$ with $1\leq v \leq n\cdot N_{\comp}$ can be answered by making one query to $S$. %
	
	Note that our construction  generates a graph $G$ from a distribution $\mathcal{D}=\{G_1,G_2,\cdots\}$. We will show that if $S$ is $k$-occurrence-free, then any graph from $\mathcal{D}$ is $H$-free; if $S$ is far from being $k$-occurrence-free, then every graph in $\mathcal{D}$ is far from $H$-freeness.

	\vspace{1em}
	\noindent\textbf{Preserving the distances.} Note that in the above construction, if there exists some value occurring $k$ times in $S$, then these $k$ occurrences of the same value results in $k$ different source components covering $\{C_1,C_2,\cdots,C_k\}$, and they are adjacent to the same center. That is, each element occurring $k$ times in the sequence result in an occurrence of $H$ in $G$. %
	For each element occurring less than $k$ times, the center corresponding to this value will be adjacent to less than $k$ source components, which in turn implies that $H$ does not occur in this case. We mention that the auxiliary isolated vertices also do not contribute to any  occurrence of $H$.
	
	Thus, if $S$ is $k$-occurrence-free, then there can not be any occurrence of $H$, and thus $G$ must be $H$-free. If $S$ is $\varepsilon$-far from being $k$-occurrence-free, then there will be at least $\varepsilon n$ occurrences of $H$ in $G$. This implies that $G$ is at least $\varepsilon'$-far from $H$-freeness, for $\varepsilon'=\frac{\varepsilon n}{d(N_{\ctr}+N_{\comp})n} = \frac{\varepsilon}{d(N_{\ctr}+N_{\comp})}$. 
	
	\vspace{1em}
	\noindent\textbf{Putting things together.} 
	Let $\mathcal{A}'$ be an algorithm for testing $H$-freeness with proximity parameter $\varepsilon=\Theta_{k,d}(1)$. %
	Suppose that the query complexity is $o(n'^{1-\frac{1}{k}})$ on an $n'$-vertex digraph. Now we invoke the algorithm $\mathcal{A}'$ on the graph $G$ that was constructed as before. As we have seen, each query in $G$ can be answered by making at most $1$ query to the sequence $S$. Furthermore, if $S$ is $k$-occurrence-free, then $G$ is $H$-free and if $S$ is $\varepsilon$-far from being $k$-occurrence-free, then $G$ is $\varepsilon'$-far from $H$-free, for $\varepsilon'=\frac{\varepsilon}{d\cdot (N_{\ctr}+N_{\comp})}=\Theta_{k,d}(1)$. Thus, the algorithm $\mathcal{A'}$, together with the construction, also solves the problem of testing $k$-occurrence-freeness with $o(n^{1-\frac{1}{k}})$ queries, which contradicts Theorem \ref{thm:kvalue-lowerbound}. Thus, the query complexity of $\mathcal{A}'$ is $\Omega(n^{1-\frac{1}{k}})$. This finishes the proof of the theorem.

\end{proof}

\bibliographystyle{alpha}
\bibliography{testing}

\end{document}